\documentclass[12pt]{article}

\usepackage[margin=1in]{geometry}
\usepackage{amsthm}
\usepackage{amssymb}
\usepackage{amsmath}
\usepackage{graphicx}
\usepackage{url}
\usepackage{color}
\usepackage{framed}
\usepackage[table]{xcolor}

\newtheorem{theorem}{Theorem}

\newtheorem{lemma}[theorem]{Lemma}
\newtheorem{proposition}[theorem]{Proposition}

\theoremstyle{definition}

\title{Kesten--McKay law for random subensembles\\of Paley equiangular tight frames}

\author{Mark~Magsino \and Dustin~G.~Mixon \and Hans Parshall}
\date{}

\begin{document}
\maketitle

\begin{abstract}
We apply the method of moments to prove a recent conjecture of Haikin, Zamir and Gavish~\cite{HaikinZG:17} concerning the distribution of the singular values of random subensembles of Paley equiangular tight frames.
Our analysis applies more generally to real equiangular tight frames of redundancy $2$, and we suspect similar ideas will eventually produce more general results for arbitrary choices of redundancy.
\end{abstract}

\section{Introduction}

Frame theory concerns redundant representation in a Hilbert space.
A \textbf{frame}~\cite{DuffinS:52} is a sequence $\{\varphi_i\}_{i\in I}$ in a Hilbert space $H$ for which there exist $\alpha,\beta\in(0,\infty)$ such that
\[
\alpha\|x\|^2
\leq\sum_{i\in I}|\langle x,\varphi_i\rangle|^2
\leq\beta\|x\|^2
\]
for every $x\in H$.
If every $\varphi_i$ has unit norm, then we say the frame is \textbf{unit norm}, and if $\alpha=\beta$, we say the frame is \textbf{tight}~\cite{DaubechiesGM:86}.
In the special case where $H=\mathbb{R}^d$, a frame is simply a spanning set, but unit norm tight frames are still interesting and useful~\cite{BenedettoF:03,Mixon:16}.
For example, \textbf{equiangular tight frames} are unit norm tight frames with the additional property that $|\langle \varphi_i,\varphi_j\rangle|$ is constant over the choice of pair $\{i,j\}$.
Equiangular tight frames are important because they necessarily span optimally packed lines, which in turn find applications in multiple description coding~\cite{StrohmerH:03}, digital fingerprinting~\cite{MixonQKF:13}, compressed sensing~\cite{BandeiraFMW:13}, and quantum state tomography~\cite{RenesBSC:04}; see~\cite{FickusM:15} for a survey.

Various applications demand control over the singular values of subensembles of frames.
In quantum physics, Weaver's conjecture~\cite{Weaver:04} (equivalent to the Kadison--Singer problem~\cite{KadisonS:59,CasazzaFTW:06}, and recently resolved in~\cite{MarcusSS:15}) concerns the existence of subensembles of unit norm tight frames with appropriately small spectral norm.
Compressed sensing~\cite{CandesRT:06,Donoho:06} has spurred the pursuit of explicit frames with the property that every subensemble is well conditioned~\cite{DeVore:07,BourgainDFKK:11,BandeiraFMW:13}.
Motivated by applications in erasure-robust analog coding, Haikin, Zamir and Gavish~\cite{HaikinZG:17,HaikinZG:18} recently launched a new line of inquiry:\ identify frames for which the singular values of random subensembles exhibit a predictable distribution.
(One might consider this to be a more detailed analogue to Tropp's estimates on the conditioning of random subensembles~\cite{Tropp:08}.)
Of particular interest are random subensembles of equiangular tight frames, and in this paper, we consider equiangular tight frames comprised of $2d$ vectors in $\mathbb{R}^d$, which correspond to symmetric conference matrices.
(Note that such frames have already received some attention in the context of compressed sensing~\cite{BandeiraFMW:13,BandeiraMM:17}.)

An $n\times n$ matrix $S$ is said to be a \textbf{conference matrix} if
\begin{itemize}
\item[(i)] $S_{ii}=0$ for every $i\in[n]$,
\item[(ii)] $S_{ij}\in\{\pm1\}$ for every $i,j\in[n]$ with $i\neq j$, and
\item[(iii)] $S^\top S=(n-1)I$.
\end{itemize}
A symmetric conference matrix of order $n$ exists whenever $n-1\equiv 1\bmod 4$ is a prime power (by a Paley--based construction), and only if $n\equiv 2\bmod 4$ and $n-1$ is a sum of two squares~\cite{IoninK:07}.
Explicitly, the Paley conference matrices are obtained by building a circulant matrix from the Legendre symbol and then padding with ones, for example:
\[
\{(\tfrac{x}{5})\}_{x=0}^4=(0,+,-,-,+)
\qquad
\Longrightarrow
\qquad
S=\left[\begin{array}{c|ccccc}0&+&+&+&+&+\\\hline+&\cellcolor{black!10}0&\cellcolor{black!10}+&\cellcolor{black!10}-&\cellcolor{black!10}-&\cellcolor{black!10}+\\+&+&0&+&-&-\\+&-&+&0&+&-\\+&-&-&+&0&+\\+&+&-&-&+&0\end{array}\right]
\]
where ``$\pm$'' denotes $\pm1$.
One may verify that the above example satisfies $S^2=5I$.
For every $n\times n$ symmetric conference matrix $S$, it holds that $I+\frac{1}{\sqrt{n-1}}S$ is the Gram matrix of an equiangular tight frame consisting of $n$ vectors in $\mathbb{R}^{n/2}$~\cite{StrohmerH:03}.
In particular, the equiangular tight frames that arise from the Paley conference matrices are known as Paley equiangular tight frames.
In what follows, we consider random principal submatrices of symmetric conference matrices with the understanding that they may be identified with the Gram matrix of a random subensemble of the corresponding equiangular tight frame.

Given an $n\times n$ symmetric matrix $Z$ with eigenvalues $\lambda_1\leq\cdots\leq\lambda_n$, we let $\mu_Z$ denote the uniform probability measure over the spectrum of $Z$ (counted with multiplicity):
\[
\mu_Z:=\frac{1}{n}\sum_{i=1}^n\delta_{\lambda_i}.
\]
This is known as the \textbf{empirical spectral distribution} of $Z$.
If $Z$ is a random matrix, then its empirical spectral distribution $\mu_Z$ is a random measure.
We say a sequence $\{\zeta_i\}_{i=1}^\infty$ of random measures \textbf{converges almost surely} to a non-random absolutely continuous measure $\mu$ if for every $a,b\in\mathbb{R}$ with $a<b$, it holds that the random variable $\zeta_i(a,b)$ converges to $\mu(a,b)$ almost surely.

We are interested in random matrices of a particular form.
Let $\mathcal{I}$ denote a random subset of $[n]$ such that the events $\{1\in\mathcal{I}\},\ldots,\{n\in \mathcal{I}\}$ are independent with probability $p$.
Then for any fixed $n\times n$ matrix $A$, we write $X\sim\operatorname{Sub}(A,p)$ to denote the (random) principal submatrix of $A$ with rows and columns indexed by $\mathcal{I}$.
Following~\cite{DubbsE:15}, we define the \textbf{Kesten--McKay distribution} with parameter $v\geq 2$ by
\[
d\mu_{\operatorname{KM}(v)}=\left\{\begin{array}{cl}\frac{v\sqrt{4(v-1)-x^2}}{2\pi(v^2-x^2)}&\text{if }x^2\leq 4(v-1)\\0&\text{otherwise}\end{array}\right\}dx.
\]
Recall that a \textbf{lacunary sequence} is a set $\{n_i:i\in\mathbb{N}\}$ of natural numbers for which there exists $\lambda>1$ such that $n_{i+1}\geq \lambda n_i$ for every $i$.
We are now ready to state our main result, which corresponds to one of many conjectures posed in~\cite{HaikinZG:17}; see Figure~\ref{figure} for an illustration.

\begin{theorem}
\label{thm.main result}
Fix $p\in(0,\frac{1}{2})$, take any lacunary sequence $L$ for which there exists a sequence $\{S_n\}_{n\in L}$ of symmetric conference matrices of increasing size $n$, and consider the corresponding random matrices $X_n\sim\operatorname{Sub}(S_n,p)$.
Then the empirical spectral distribution of $\frac{1}{p\sqrt{n}}X_n$ converges almost surely to the Kesten--McKay distribution with parameter $v=1/p$.
\end{theorem}

\begin{figure}
\begin{center}
\includegraphics[width=\textwidth]{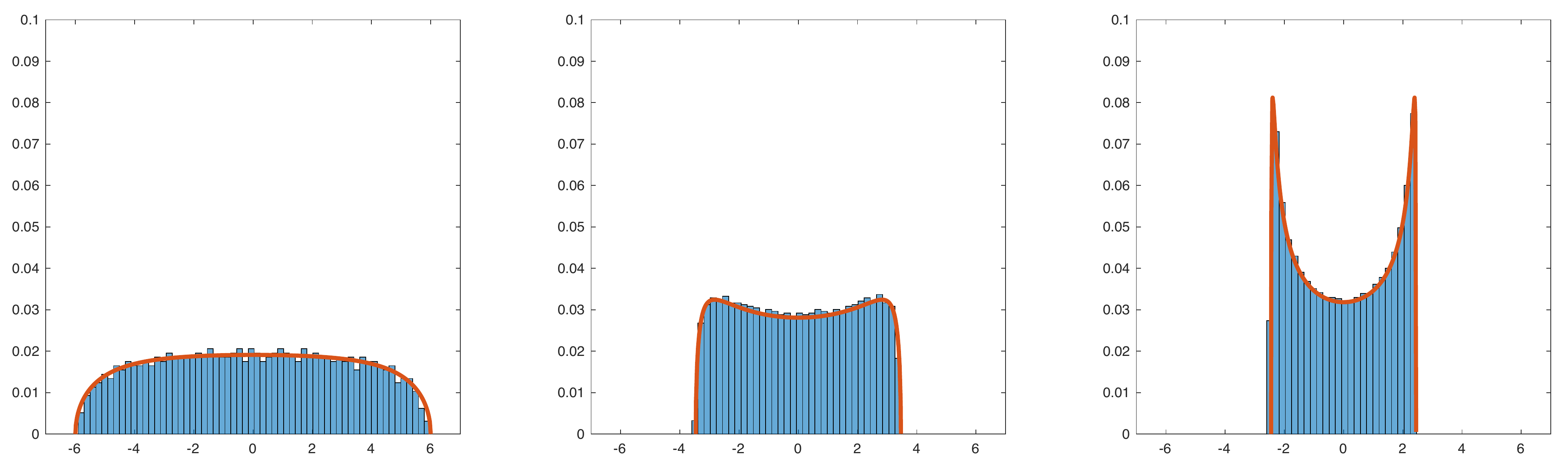}
\end{center}
\caption{\label{figure}
Consider the Paley conference matrix $S$ of order $n=10,010$.
For each choice of $p\in\{0.1,0.25,0.4\}$, we draw $X\sim\operatorname{Sub}(S,p)$ and plot a histogram of the spectrum of $\frac{1}{p\sqrt{n}}X$ along with a suitably scaled version of the Kesten--McKay density for $v=1/p$.
The similarity between these distributions was first observed by Haikin, Zamir and Gavish~\cite{HaikinZG:17}.
Our main result (Theorem~\ref{thm.main result}) explains this phenomenon. 
}
\end{figure}

In the next section, we prove this theorem using the method of moments, saving the more technical portions for Section~3.

\subsection{Notation}

Given $x\in\mathbb{R}^n$, let $\operatorname{diag}(x)$ denote the $n\times n$ diagonal matrix whose diagonal entries are the entries of $x$.
Given $Z\in\mathbb{R}^{m\times n}$, let $\|Z\|_{2\to2}$ denote the induced $2$-norm of $Z$ (i.e., the largest singular value of $Z$), and let $\|Z\|_{S^p}$ denote the Schatten $p$-norm of $Z$ (i.e., the $p$-norm of the singular values of $Z$).
Throughout this paper, we will investigate how quantities relate as $n\to\infty$.
For example, suppose we are interested in a quantity $f(n,\theta)\geq0$ that depends on both $n\in\mathbb{N}$ and some additional parameters $\theta\in\mathbb{R}^m$.
Then we write $f(n,\theta)=o(g(n,\theta))$ if for every $\theta\in\mathbb{R}^m$, it holds that $f(n,\theta)/g(n,\theta)\to 0$ as $n\to\infty$.
We write $f(n,\theta)\lesssim g(n,\theta)$ if there exists $c>0$ such that $f(n,\theta)\leq c\cdot g(n,\theta)$ for all $n\in\mathbb{N}$ and $\theta\in\mathbb{R}^m$, and we write $f(n,\theta)\lesssim_\theta g(n,\theta)$ if for every $\theta\in\mathbb{R}^m$, there exists $c(\theta)>0$ such that $f(n,\theta)\leq c(\theta)\cdot g(n,\theta)$ for all $n\in\mathbb{N}$.
Finally, we write $f(n,\theta)\asymp g(n,\theta)$ if both $f(n,\theta)\lesssim g(n,\theta)$ and $g(n,\theta)\lesssim f(n,\theta)$.

\section{Proof of the main result}

Our proof makes use of a standard sufficient condition for the almost sure convergence of random measures, which is a consequence of the moment continuity theorem, the Borel--Cantelli lemma, and Chebyshev's inequality, cf.\ Exercise~2.4.6 in~\cite{Tao:11}:

\begin{proposition}
\label{prop.key proposition}
Let $\{\zeta_i\}_{i=1}^\infty$ be a sequence of uniformly subgaussian random probability measures, and let $\mu$ be a non-random subgaussian probability measure.
Suppose that for every $k\in\mathbb{N}$, it holds that
\begin{itemize}
\item[(i)]
$\displaystyle\mathbb{E}\int_\mathbb{R}x^k d\zeta_i(x)\to\int_\mathbb{R}x^kd\mu(x)$, and
\item[(ii)]
$\displaystyle\sum_{i=1}^\infty\operatorname{Var}\bigg(\int_\mathbb{R}x^kd\zeta_i(x)\bigg)<\infty$.
\end{itemize}
Then $\zeta_i$ converges almost surely to $\mu$.
\end{proposition}

As we will see, verifying hypothesis~(i) in our case reduces to a combinatorics problem, whereas hypothesis~(ii) can be treated separately with the help of Talagrand concentration:

\begin{proposition}[Talagrand concentration, Theorem~2.1.13 in~\cite{Tao:11}]
\label{prop.talagrand}
There exists a universal constant $c>0$ for which the following holds:
Suppose $f\colon\mathbb{R}^n\to\mathbb{R}$ is both convex and $\sigma$-Lipschitz in $\|\cdot\|_2$, and let $X$ be a random vector in $\mathbb{R}^n$ with independent coordinates satisfying $|X_i|\leq b$ almost surely.
Then for every $t\geq 0$, it holds that
\[
\mathbb{P}\big\{|f(X)-\mathbb{E}f(X)|\geq bt\big\}
\lesssim e^{-t^2/c\sigma^2}.
\]
\end{proposition}

Throughout, $S_n$ denotes an $n\times n$ symmetric conference matrix, we draw $X_n\sim\operatorname{Sub}(S_n,p)$ and put $Z_n:=\frac{1}{p\sqrt{n}}X_n$.
We typically suppress the subscript $n$.  
While the size of $Z$ is random, its average size is $pn$, and so we use $\frac{1}{pn}\operatorname{tr}(Z^k)$ as a proxy for $\int_\mathbb{R}x^kd\mu_Z(x)$.
As one might expect, this is a good approximation:

\begin{lemma}
\label{lem.moment approximation}
Put $V:=\frac{1}{pn}\operatorname{tr}(Z^k)$ and $W:=\int_\mathbb{R}x^kd\mu_{Z}(x)$.
Then
\[
|\mathbb{E}V-\mathbb{E}W|
\lesssim_p\frac{1}{\sqrt{n}},
\qquad
|\operatorname{Var}(V)-\operatorname{Var}(W)|
\lesssim_p\frac{1}{\sqrt{n}}.
\]
\end{lemma}

\begin{proof}
Since $X$ is a submatrix of $S$, it holds that
\[
|V|
\lesssim_p\frac{1}{n}\sum_i|\lambda_i(Z)|^k
\leq\|Z\|_{2\to2}^k
\asymp_p\frac{1}{n^{k/2}}\|X\|_{2\to2}^k
\leq \frac{1}{n^{k/2}}\|S\|_{2\to2}^k
\leq 1
\]
almost surely.
Similarly, $|W|\leq\|Z\|_{2\to2}^k\lesssim_p 1$ almost surely.
Next, let $N$ denote the (random) size of $Z$.
Then $V=\frac{N}{pn}\cdot W$, and so our bound on $|W|$ gives
\[
\mathbb{E}|V-W|
=\mathbb{E}\big(|\tfrac{N}{pn}-1|\cdot|W|\big)
\lesssim_p\mathbb{E}|\tfrac{N}{pn}-1|
\leq\frac{1}{pn}\big(\mathbb{E}(N-pn)^2\big)^{1/2}
\lesssim_p\frac{1}{\sqrt{n}},
\]
where the last step applies the fact that $N$ has binomial distribution.
This immediately implies the desired bound on $|\mathbb{E}V-\mathbb{E}W|$.
Finally, since $|V|,|W|\lesssim_p 1$ almost surely, we have
\begin{align*}
|\operatorname{Var}(V)-\operatorname{Var}(W)|
&\leq|\mathbb{E}V^2-\mathbb{E}W^2|+|(\mathbb{E}V)^2-(\mathbb{E}W)^2|\\
&\leq\mathbb{E}\big(|V+W||V-W|\big)+|\mathbb{E}V+\mathbb{E}W||\mathbb{E}V-\mathbb{E}W|\\
&\lesssim_p\mathbb{E}|V-W|
\lesssim_p\frac{1}{\sqrt{n}},
\end{align*}
which completes the result.
\end{proof}

As such, to demonstrate hypothesis~(i) from Proposition~\ref{prop.key proposition} in our case, it suffices to prove
\begin{equation}
\label{eq.desired moments}
\mathbb{E}\frac{1}{pn}\operatorname{tr}(Z^k)
\to\int_\mathbb{R}x^k~d\mu_{\operatorname{KM}(1/p)}(x).
\end{equation}
The Kesten--McKay moments are implicitly computed in~\cite{McKay:81}, and are naturally expressed in terms of entries of \textbf{Catalan's triangle}:
\[
C(n,k)
:=\frac{(n+k)!(n-k+1)}{k!(n+1)!}.
\]

\begin{proposition}[Lemma~2.1 in~\cite{McKay:81}]
\label{prop.km moments}
For every $v\geq2$ and $k\in\mathbb{N}$, it holds that
\[
\int_\mathbb{R}x^k~d\mu_{\operatorname{KM}(v)}(x)
=\left\{\begin{array}{cl}\displaystyle
\sum_{j=1}^{k/2}C(k/2-1,k/2-j)v^j(v-1)^{k/2-j}&\text{if $k$ is even}\\0&\text{if $k$ is odd.}\end{array}\right.
\]
\end{proposition}

Recalling that $Z=\frac{1}{p\sqrt{n}}X$, then Proposition~\ref{prop.km moments} gives that \eqref{eq.desired moments} is equivalent to 
\begin{equation}
\label{eq.limiting expression}
\frac{1}{n^{k/2+1}}\mathbb{E}\operatorname{tr}(X^k)
\to\left\{\begin{array}{cl}
\displaystyle\sum_{t=k/2+1}^{k}(-1)^{t-k/2-1}\cdot B(k/2-1,t-k/2-1)\cdot p^t&\text{if $k$ is even}\\
0&\text{if $k$ is odd,}
\end{array}\right.
\end{equation}
where $B(n,k)$ denotes an entry of \textbf{Borel's triangle}:
\[
B(n,k)
:=\sum_{j=k}^n\binom{j}{k}C(n,j).
\]
To compute these limits, we first find a convenient expression for $\frac{1}{n^{k/2+1}}\mathbb{E}\operatorname{tr}(X^k)$.
To this end, recall that $X$ is the submatrix of $S$ with index set $\mathcal{I}$, and let $P$ denote the random $n\times n$ diagonal matrix such that $P_{ii}=1_{\{i\in\mathcal{I}\}}$.
Then 
\begin{align*}
\operatorname{tr}(X^k)
=\operatorname{tr}((PSP)^k)
=\operatorname{tr}((PS)^k)
&=\sum_{a_1,\ldots,a_k\in[n]}(PS)_{a_1a_2}(PS)_{a_2a_3}\cdots(PS)_{a_ka_1}\\
&=\sum_{a_1,\ldots,a_k\in[n]}S_{a_1a_2}S_{a_2a_3}\cdots S_{a_ka_1}\cdot\prod_{i=1}^k1_{\{a_i\in\mathcal{I}\}}.
\end{align*}
Considering $\mathbb{E}\prod_{i=1}^k1_{\{a_i\in\mathcal{I}\}}=p^{|\{a_1,\ldots,a_k\}|}$, it follows that
\begin{equation}
\label{eq.rewrite non-asymptotic expectation}
\frac{1}{n^{k/2+1}}\mathbb{E}\operatorname{tr}(X^k)
=\sum_{t=1}^k \bigg(\frac{1}{n^{k/2+1}}\sum_{\substack{a_1,\ldots,a_k\in[n]\\|\{a_1,\ldots,a_k\}|=t}}S_{a_1a_2}S_{a_2a_3}\cdots S_{a_ka_1}\bigg) \cdot p^t.
\end{equation}
It remains to show that these coefficients converge to the corresponding coefficients in \eqref{eq.limiting expression}.

First, we introduce some additional notation.
Taking inspiration from Bargmann invariants~\cite{Bargmann:64}, it is convenient to write
\[
\Delta(a_1,a_2,a_3,\ldots,a_k)
:=S_{a_1a_2}S_{a_2a_3}\cdots S_{a_ka_1}.
\]
Next, we say $\pi$ is a \textbf{partition} of $[k]$ into $t$ blocks if $\pi=\{B_1,\ldots,B_t\}$ such that $B_1\sqcup\cdots\sqcup B_t=[k]$, and we let $\Pi(k,t)$ denote the set of all such partitions.
For each partition $\pi$ of $[k]$, we consider the set of functions $a\colon[k]\to[n]$ whose level sets are the blocks of $\pi$, namely
\[
L_n(\pi)
:=\big\{a\colon[k]\to[n]:\{a^{-1}(a(i)):i\in[k]\}=\pi\big\}.
\]
With this, we define
\[
V_n(\pi)
:=\frac{1}{n^{k/2+1}}\sum_{a\in L_n(\pi)}\Delta(a(1),\ldots,a(k)).
\]
Considering \eqref{eq.rewrite non-asymptotic expectation}, it therefore holds that
\begin{equation}
\label{eq.rewrite non-asymptotic expectation again}
\frac{1}{n^{k/2+1}}\mathbb{E}\operatorname{tr}(X^k)
=\sum_{t=1}^k\bigg(\sum_{\pi\in\Pi(k,t)}V_n(\pi)\bigg)\cdot p^t.
\end{equation}
As such, to demonstrate~\eqref{eq.limiting expression}, it suffices to determine the limit of $V_n(\pi)$ for every partition $\pi$ of $[k]$.
We start with a quick calculation:

\begin{lemma}
\label{fewblocksLemma}
For every $\pi\in\Pi(k,t)$ with $t < k/2 + 1$, it holds that $V_n(\pi)\to0$.
\end{lemma}

\begin{proof}
Estimate $|V_n(\pi)|$ using the triangle inequality to obtain a sum of $|L_n(\pi)| \leq n^t=o(n^{k/2+1})$ terms, each of size at most $1$.
\end{proof}

For each $t< k/2+1$, this establishes that the coefficient of $p^t$ in \eqref{eq.rewrite non-asymptotic expectation again} approaches zero, i.e., the corresponding coefficient in \eqref{eq.limiting expression}.
Now we wish to tackle the limiting value of $V_n(\pi)$ in general.
In light of the related literature~\cite{NicaS:06}, it comes as no surprise that $V_n(\pi)$ depends on whether $\pi$ is a so-called \textit{crossing partition}.
We say a partition $\pi$ of $[k]$ is \textbf{crossing} if there exist $A,B\in\pi$ with $A\neq B$ for which there exist $a_1,a_2\in A$ and $b_1,b_2\in B$ such that $a_1<b_1<a_2<b_2$.
Otherwise, $\pi$ is said to be \textbf{non-crossing}.
Next, for each $x\in[k]$, we let $\pi(x)$ denote the unique member of $\pi$ such that $x\in\pi(x)$.
Consider the graph $G_\pi$ with vertex set $\pi$ and edges given by $\pi(x)\leftrightarrow\pi(x+1)$ for every $x\in[k]$; here, we interpret $x+1$ modulo $k$ so that $k+1=1$.
Let $\operatorname{EC}(k,t)$ denote the set of non-crossing $\pi\in\Pi(k,t)$ for which the edges of $G_\pi$ partition into simple even cycles.
Finally, let $C_n:=\frac{1}{n+1}\binom{2n}{n}$ denote the $n$th \textbf{Catalan number}.
With these notions, we can describe the limit of each $V_n(\pi)$:

\begin{lemma}[Key combinatorial lemma]\
\label{lem.key lemma}
\begin{itemize}
\item[(i)]
Suppose $\pi\in\Pi(k,t)\setminus\operatorname{EC}(k,t)$.
Then $V_n(\pi)\to0$.
\item[(ii)]
Suppose $\pi\in\operatorname{EC}(k,t)$ and the edges of $G_\pi$ partition into $m$ simple cycles of sizes $2s_1,\ldots,2s_m$.
Then $m=k-t+1$ and
\[
V_n(\pi)
\to(-1)^{k/2-m}\cdot C_{s_1-1}\cdots C_{s_m-1}.
\]
\end{itemize}
\end{lemma}

The proof of Lemma~\ref{lem.key lemma} is rather technical (involving multiple rounds of induction), and so we save it for Section~3.
In the meantime, we demonstrate how Lemma~\ref{lem.key lemma} can be applied to prove that the coefficients in \eqref{eq.rewrite non-asymptotic expectation again} converge to the coefficients in \eqref{eq.limiting expression}.
Recall that a \textbf{Dyck path} of semi-length $n$ is a path in the plane from $(0,0)$ to $(2n,0)$ consisting of $n$ steps along the vector $(1,1)$, called \textbf{up-steps}, and $n$ steps along the vector $(1,-1)$, called \textbf{down-steps}, that never goes below the $x$-axis.
We say a Dyck path is \textbf{strict} if none of the path's interior vertices reside on the $x$-axis.
Each (strict) Dyck path determines a sequence of $2n$ letters from $\{U,D\}$ that represent up- and down-steps in the path; this sequence is known as a \textbf{(strict) Dyck word}.
With these notions, we may prove the following result by leveraging the fact that Borel's triangle counts so-called \textit{marked Dyck paths}~\cite{CaiY:19}; see Figure~\ref{fig.dyck paths} for an illustration.

\begin{figure}
\begin{center}
\includegraphics[width=\textwidth]{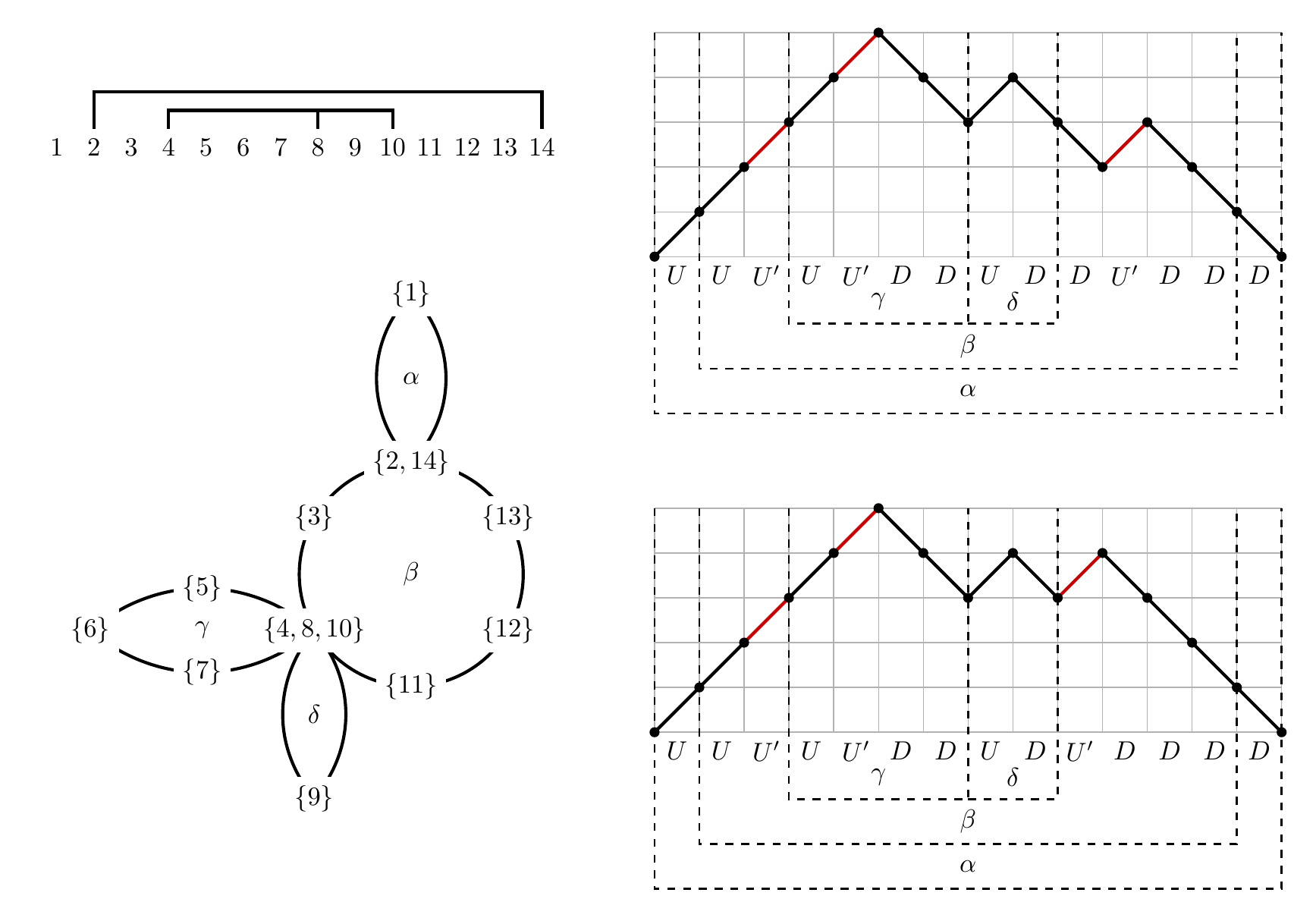}
\end{center}
\caption{\label{fig.dyck paths}
\textbf{(top left)}
Select $k=14$ and $t=11$, and consider the partition $\pi\in\Pi(k,t)$ with all singleton blocks except for $\{2,14\}$ and $\{4,8,10\}$.
Observe that $\pi$ is a non-crossing partition.
\textbf{(bottom left)}
We depict the corresponding graph $G_\pi$, whose vertices are the blocks of $\pi$.
By definition, blocks are adjacent in $G_\pi$ when they contain cyclicly adjacent members of $[k]$.
In this case, the edges of $G_\pi$ partition into four simple cycles, which we label $\alpha$, $\beta$, $\gamma$ and $\delta$.
\textbf{(right)}
Each simple cycle of $G_\pi$ is assigned a strict Dyck word of the cycle's length, and we mark all but the first up-steps.
The only choice for $\alpha$ and $\delta$ is $UD$, and the only choice for $\gamma$ is $UU'DD$; here, $U'$ denotes a marked up-step.
Meanwhile, $\beta$ has $C_2=2$ choices: $UU'DU'DD$ and $UU'U'DDD$.
For each selection, we traverse $G_\pi$ from $\pi(1)$ to $\pi(2)$, to $\pi(3)$, etc., to $\pi(14)$ and back to $\pi(1)$, labelling the edges of $G_\pi$ with the next letter from the current cycle's Dyck word.
The result is a Dyck word with $t-k/2-1=3$ marked up-steps, none of which at ground level.
We illustrate the corresponding marked Dyck paths above.
Notice that $G_\pi$ can be recovered from either marked Dyck path since a cycle is born with each un-marked up-step and dies once the Dyck path returns to its height from the birth of that cycle.
By Theorem~2 in~\cite{CaiY:19}, marked Dyck paths are counted by entries in Borel's triangle, which explains their appearance in~\eqref{eq.limiting expression}.
}
\end{figure}

\begin{lemma}
\label{lem.counting with borel}
It holds that
\[
\displaystyle\sum_{\pi\in\Pi(k,t)}V_n(\pi)
\to\left\{\begin{array}{cl}(-1)^{t-k/2-1}\cdot B(k/2-1,t-k/2-1)&\text{if $k$ is even and $t \geq k/2 + 1$}\\0&\text{otherwise.}\end{array}\right.
\]
\end{lemma}

\begin{proof}
When $t<  k/2 + 1$, the result follows from Lemma~\ref{fewblocksLemma}, and when $k$ is odd, the $k$ edges in each $G_\pi$ fail to partition into even simple cycles, and so the result follows from Lemma~\ref{lem.key lemma}(i).
Now suppose $k$ is even and $t \geq k/2 + 1$.
For $\pi\in\operatorname{EC}(k,t)$, recall that the edges of $G_\pi$ are indexed by $[k]$ and partitioned into simple even cycles.
Define $\operatorname{MD}(\pi)$ to be the words $w\colon [k]\to\{U,U',D\}$ such that for every simple cycle in $G_\pi$ with edges indexed by $T\subseteq[k]$, the restriction $w|_T$ is a strict Dyck word with all but its first up-steps marked (here, $U'$ denotes a marked up-step).
Note that strict Dyck words of semi-length $s$ are in one-to-one correspondence with Dyck words of semi-length $s-1$, and so there are $C_{s-1}$ of them.
As such, Lemma~\ref{lem.key lemma} implies that for every $\pi\in\Pi(k,t)$, it holds that
\begin{equation}
\label{eq.value to marked dyke}
(-1)^{t-k/2-1}\cdot V_n(\pi)
\to\left\{\begin{array}{cl}|\operatorname{MD}(\pi)|&\text{if }\pi\in\operatorname{EC}(k,t)\\0&\text{otherwise.}\end{array}\right.
\end{equation}
Let $\operatorname{MD}(k,t)$ denote the set of marked Dyck words $w\colon [k]\to\{U,U',D\}$ with $t-k/2-1$ marked up-steps, none of which are at ground level.
We observe that
\begin{equation}
\label{eq.partition of marked dyke}
\operatorname{MD}(k,t)
=\bigsqcup_{\pi\in\operatorname{EC}(k,t)}\operatorname{MD}(\pi).
\end{equation}
Then equations \eqref{eq.value to marked dyke} and \eqref{eq.partition of marked dyke} together give
\[
(-1)^{t-k/2-1}\sum_{\pi\in\Pi(k,t)}V_n(\pi)
\to\sum_{\pi\in\operatorname{EC}(k,t)}|\operatorname{MD}(\pi)|
=|\operatorname{MD}(k,t)|
=B(k/2-1,t-k/2-1),
\]
where the last step applies Theorem~2 in~\cite{CaiY:19}.
\end{proof}

At this point, we are in a position to verify hypothesis~(i) from Proposition~\ref{prop.key proposition} in our case.
For hypothesis~(ii), we follow the approach suggested by Remark~2.4.5 in~\cite{Tao:11} of leveraging Talagrand concentration to bound the variance.
First, we pass to a setting that is more amenable to analysis with Talagrand concentration.
Here and throughout, for each $n\in L$, we fix an $n\times n$ matrix $F$ such that $F^\top F=I+\frac{1}{\sqrt{n}}S_n$.

\begin{lemma}
\label{lem.pass to Talagrand-friendly}
It holds that $\displaystyle \operatorname{Var}\big(\operatorname{tr}\big((\tfrac{1}{p\sqrt{n}}X)^k\big)\big)
\lesssim_{p,k} \max_{j\in[k]}\operatorname{Var}\big(\|FP\|_{S^{2j}}^{2j}\big)+n^{1/2}$.
\end{lemma}

\begin{proof}
Define $Y:=\frac{1}{p\sqrt{n}}PSP$, and observe that
\[
\operatorname{tr}(Y^k)
=\operatorname{tr}(\tfrac{1}{p\sqrt{n}}X^k),
\qquad
\operatorname{tr}\big((Y+\tfrac{1}{p}P)^k\big)
=\operatorname{tr}\big((\tfrac{1}{p}PF^\top FP)^k\big)
=\frac{1}{p^k}\|FP\|_{S^{2k}}^{2k}.
\]
Since $Y$ commutes with $P$ and $YP=Y$, the binomial theorem gives
\[
\operatorname{tr}\big((Y+\tfrac{1}{p}P)^k\big)
=\operatorname{tr}\sum_{j=0}^k\binom{k}{j}Y^j(\tfrac{1}{p}P)^{k-j}
=\sum_{j=1}^k\binom{k}{j}\frac{1}{p^{k-j}}\operatorname{tr}(Y^j)+\frac{1}{p^k}\operatorname{tr}(P),
\]
and so rearranging gives
\begin{align}
\nonumber
\operatorname{tr}(\tfrac{1}{p\sqrt{n}}X^k)
=\operatorname{tr}(Y^k)
&=\operatorname{tr}\big((Y+\tfrac{1}{p}P)^k\big)-\sum_{j=1}^{k-1}\binom{k}{j}\frac{1}{p^{k-j}}\operatorname{tr}(Y^j)-\frac{1}{p^k}\operatorname{tr}(P)\\
\label{eq.eigenshift identity}&=\frac{1}{p^k}\|FP\|_{S^k}^{2k}-\sum_{j=1}^{k-1}\binom{k}{j}\frac{1}{p^{k-j}}\operatorname{tr}(\tfrac{1}{p\sqrt{n}}X^j)-\frac{1}{p^k}\operatorname{tr}(P).
\end{align}
The following estimate holds for any choice of random variables $\{X_i\}_{i\in[m]}$:
\begin{align*}
\operatorname{Var}\bigg(\sum_{i=1}^m X_i\bigg)
=\sum_{i=1}^m\sum_{j=1}^m\operatorname{Cov}(X_i,X_j)
&\leq\sum_{i=1}^m\sum_{j=1}^m|\operatorname{Cov}(X_i,X_j)|\\
&\leq\sum_{i=1}^m\sum_{j=1}^m\sqrt{\operatorname{Var}(X_i)\operatorname{Var}(X_j)}
\leq m^2\cdot\max_{i\in[m]}\operatorname{Var}(X_i).
\end{align*}
The lemma follows from applying this estimate to \eqref{eq.eigenshift identity} by induction on $k$.
\end{proof}

Next, we establish the convexity and Lipschitz continuity required by Talagrand:

\begin{lemma}
\label{lem.convex lipschitz}
For each $k\in\mathbb{N}$, consider the mapping $f\colon \{x\in\mathbb{R}^n:\|x\|_\infty<2\}\to\mathbb{R}$ defined by $f(x)=\|F\operatorname{diag}(x)\|_{S^{2k}}^{2k}$.
Then $f$ is convex and $(8^kkn^{1-1/2k})$-Lipschitz.
\end{lemma}

\begin{proof}
We adopt the shorthand notation $D_x:=\operatorname{diag}(x)$.
First, $f$ is convex since $\|\cdot\|_{S^{2k}}$ satisfies the triangle inequality and $t\mapsto t^{2k}$ is convex:
\[
f\big(\lambda x+(1-\lambda)y\big)
\leq \Big(\lambda\|FD_x\|_{S^{2k}}+(1-\lambda)\|FD_y\|_{S^{2k}}\Big)^{2k}
\leq \lambda f(x)+(1-\lambda) f(y).
\]
To compute a Lipschitz bound, we apply the factorization
\[
u^{2k}-v^{2k}
=(u-v)(u+v)\sum_{j=0}^{k-1}u^{2(k-1-j)}v^{2j}
\]
with $u:=\|FD_x\|_{S^{2k}}$ and $v:=\|FD_y\|_{S^{2k}}$ to get
\[
\big|f(x)-f(y)\big|
=|u^{2k}-v^{2k}|
=\bigg((u+v)\sum_{j=0}^{k-1}u^{2(k-1-j)}v^{2j}\bigg)\cdot|u-v|
\leq 8^kkn^{1-1/2k}\cdot|u-v|,
\]
where the last step follows from the fact that $\|FD_x\|_{2\to2}\leq 2\|F\|_{2\to2}\leq 2\sqrt{2}$, meaning $u\leq 2\sqrt{2}n^{1/2k}$ (and similarly for $v$).
Next, we apply the reverse triangle inequality to get
\begin{align*}
|u-v|
&=\big|\|FD_x\|_{S^{2k}}-\|FD_y\|_{S^{2k}}\big|\\
&\leq\|FD_x-FD_y\|_{S^{2k}}
\leq\|F(D_x-D_y)\|_F
\leq\|F\|_{2\to2}\cdot\|x-y\|_2
\leq \sqrt{2}\cdot\|x-y\|_2,
\end{align*}
which implies the result.
\end{proof}

Finally, we apply Talagrand concentration to obtain a variance bound:

\begin{lemma}
\label{lem.variance bound}
It holds that $\displaystyle \operatorname{Var}\big(\tfrac{1}{pn}\operatorname{tr}\big((\tfrac{1}{p\sqrt{n}}X)^k\big)\big)\lesssim_{p,k} n^{-1/k}$.
\end{lemma}

\begin{proof}
Given the mapping $f$ from Lemma~\ref{lem.convex lipschitz}, define $\tilde{f}\colon\mathbb{R}^n\to\mathbb{R}$ in terms of subgradients by
\[
\tilde{f}(x)
:=\sup_{x_0\in R}\sup_{z\in\partial f(x_0)}\Big(f(x_0)+\langle z,x-x_0\rangle\Big).
\]
This is known as the smallest convex extension of $f$ to $\mathbb{R}^n$, and it is straightforward to verify that $\tilde{f}$ is convex and $(8^kkn^{1-1/2k})$-Lipschitz with $\tilde{f}|_R=f$.
Let $B\in\mathbb{R}^n$ have independent entries, each equal to $1$ with probability $p$ and $0$ otherwise.
Since $B\in R$ almost surely, it holds that $\tilde{f}(B)$ has the same distribution as $\|FP\|_{S^{2k}}^{2k}$, and we let $E$ denote its expectation.
By Talagrand concentration (Proposition~\ref{prop.talagrand}), there exists $c>0$ such that
\begin{align*}
\operatorname{Var}\big(\|FP\|_{S^{2k}}^{2k}\big)
=\mathbb{E}\Big[\big(\|FP\|_{S^{2k}}^{2k}-E\big)^2\Big]
&=\int_0^\infty\mathbb{P}\Big\{\big(\|FP\|_{S^{2k}}^{2k}-E\big)^2\geq u\Big\}du\\
&=\int_0^\infty\mathbb{P}\Big\{\big|\tilde{f}(B)-E\big|\geq\sqrt{u}\Big\}du\\
&\lesssim\int_0^\infty \operatorname{exp}\bigg(\frac{-u}{c\cdot8^{2k}k^2n^{2-1/k}}\bigg)du
=c\cdot8^{2k}k^2n^{2-1/k}.
\end{align*}
Combining with Lemma~\ref{lem.pass to Talagrand-friendly} then gives
\begin{align*}
\operatorname{Var}\big(\tfrac{1}{pn}\operatorname{tr}\big((\tfrac{1}{p\sqrt{n}}X)^k\big)\big)
&=\frac{1}{p^2n^2}\operatorname{Var}\big(\operatorname{tr}\big((\tfrac{1}{p\sqrt{n}}X)^k\big)\big)\\
&\lesssim_{p,k} \frac{1}{n^2}\Big(\max_{j\in[k]}\operatorname{Var}\big(\|FP\|_{S^{2j}}^{2j}\big)+n^{1/2}\Big)
\lesssim_k n^{-1/k},
\end{align*}
as desired.
\end{proof}

We may now verify hypotheses~(i) and~(ii) from Proposition~\ref{prop.key proposition} in our case.

\begin{proof}[Proof of Theorem~\ref{thm.main result}]
Put $Z_n:=\frac{1}{p\sqrt{n}}X_n$ and $\mu:=\mu_{\operatorname{KM}(1/p)}$.
First, we modify the random measure $\mu_{Z_n}$ so that we may apply Proposition~\ref{prop.key proposition} to prove the result.
Indeed, $\mu_{Z_n}$ fails to be a probability measure with probability $(1-p)^n$, since $\mu_{Z_n}=0$ when $\mathcal{I}=\mathcal{I}_n$ is the empty set.
To rectify this, we define
\[
\zeta_n
:=\left\{\begin{array}{cl}\mu_{Z_n}&\text{if }\mathcal{I}_n\neq\emptyset\\\delta_0&\text{otherwise.}\end{array}\right.
\]
Then it suffices to prove $\zeta_n\to\mu$ almost surely, since the Borel--Cantelli lemma implies $1_{\{\mathcal{I}_n=\emptyset\}}\to0$ almost surely, and so
\[
\mu_{Z_n}(a,b)
=\zeta_n(a,b)-1_{\{\mathcal{I}_n=\emptyset\}}\cdot1_{\{0\in(a,b)\}}
\stackrel{\text{a.s.}}{\longrightarrow} \mu(a,b)
\]
for every $a,b\in\mathbb{R}$ with $a<b$.
Conveniently, for every $n\in L$ and $k\in\mathbb{N}$, it holds that
\[
\int_\mathbb{R}x^kd\zeta_n(x)
=\int_\mathbb{R}x^kd\mu_{Z_n}(x)
\]
almost surely, and so the left-hand side inherits moments from the right-hand side.

To apply Proposition~\ref{prop.key proposition}, we first observe that 
\[
\|Z_n\|_{2\to2}=\frac{1}{p\sqrt{n}}\|X_n\|_{2\to2}\leq\frac{1}{p\sqrt{n}}\|S_n\|_{2\to2}\leq\frac{1}{p}
\]
almost surely, and so $\{\zeta_n\}_{n\in L}$ are uniformly bounded, and therefore uniformly subgaussian.
Similarly, $\mu$ is bounded and therefore subgaussian.
Fix $k\in\mathbb{N}$.
As a consequence of Lemma~\ref{lem.counting with borel}, it holds that
\[
\mathbb{E}\frac{1}{pn}\operatorname{tr}(Z_n^k)
\to\int_\mathbb{R}x^kd\mu(x),
\]
and so by Lemma~\ref{lem.moment approximation}, we have
\[
\mathbb{E}\int_\mathbb{R}x^kd\zeta_n(x)
=\mathbb{E}\int_\mathbb{R}x^kd\mu_{Z_n}(x)
\to\int_\mathbb{R}x^kd\mu(x).
\]
As such, $\{\zeta_n\}_{n\in L}$ satisfies hypothesis~(i) from Proposition~\ref{prop.key proposition}.
Next, Lemma~\ref{lem.variance bound} establishes that $\operatorname{Var}\big(\tfrac{1}{pn}\operatorname{tr}(Z_n^k)\big)\lesssim_{p,k}n^{-1/k}$, and so Lemma~\ref{lem.moment approximation} implies
\[
\operatorname{Var}\bigg(\int_\mathbb{R}x^kd\zeta_n(x)\bigg)
=\operatorname{Var}\bigg(\int_\mathbb{R}x^kd\mu_{Z_n}(x)\bigg)
\lesssim_{p,k}n^{-1/k}+n^{-1/2}
\lesssim n^{-1/(k+1)}.
\]
Writing $L=\{n_i:i\in\mathbb{N}\}$, select $\lambda>1$ such that $n_{i+1}\geq \lambda n_i$ for every $i\in\mathbb{N}$.
Then
\begin{align*}
\sum_{n\in L}\operatorname{Var}\bigg(\int_\mathbb{R}x^kd\zeta_n(x)\bigg)
&\lesssim_{p,k} \sum_{n\in L}n^{-1/(k+1)}\\
&\leq \sum_{i=0}^\infty (\lambda^i n_1)^{-1/(k+1)}
=n_1^{-1/(k+1)}\sum_{i=0}^\infty (\lambda^{-1/(k+1)})^i
<\infty.
\end{align*}
As such, $\{\zeta_n\}_{n\in L}$ also satisfies hypothesis~(ii) from Proposition~\ref{prop.key proposition}, and so $\zeta_n\to\mu$ almost surely, as desired.
\end{proof}

\section{Proof of Lemma~\ref{lem.key lemma}}

It remains to compute, for each $\pi \in \Pi(k,t)$, the limit of
\[
V_n(\pi) = \frac{1}{n^{k/2 + 1}} \sum_{a \in L_n(\pi)} \Delta(a(1), \ldots, a(k)),
\]
where $L_n(\pi)$ is the set of $a \colon [k] \rightarrow [n]$ whose level sets are the blocks of $\pi$ and
\[
\Delta(a(1), \ldots, a(k)) = S_{a(1)a(2)} S_{a(2)a(3)} \cdots S_{a(k)a(1)}.
\]
We begin with some basic properties of $\Delta$.

\begin{lemma}\label{deltaLemma}
For every $a_1, \ldots, a_k \in [n]$, each of the following holds:
\begin{enumerate}
\item[(i)]
If $a_1 \neq a_2$, then $\Delta(a_1,a_2) = 1$.
\item[(ii)]
If $a_j = a_{j + 1}$ for any $j \in [k-1]$ or $a_k = a_1$, then $\Delta(a_1, \ldots, a_k) = 0$.
\item[(iii)]
If $\sigma$ is any cyclic permutation of $[k]$, then $\Delta(a_{\sigma(1)}, \ldots, a_{\sigma(k)})=\Delta(a_1, \ldots, a_k)$.
\item[(iv)]
If $a_1 \neq a_{k - 1}$, then $\sum_{b \in [n]} \Delta(a_1, \ldots, a_{k-1},b) = 0$.
\item[(v)]
If $a_1 = a_{k - 1}$ and $a_1 \neq a_k$, then $\Delta(a_1, \ldots, a_k) = \Delta(a_1, \ldots, a_{k - 2})$.
\end{enumerate}
\end{lemma}
\begin{proof}
First, (i) follows from the fact that $S$ is symmetric with off-diagonal entries in $\{\pm 1\}$.
Next, (ii) follows from the fact that the diagonal entries of $S$ are 0.
Recalling the definition of $\Delta$, then (iii) follows from commutativity.
Next suppose $a_{k - 1} \neq a_1$.  Then
\[
\sum_{b \in [n]} \Delta(a_1, \ldots, a_{k-1},b)
= S_{a_1a_2} \cdots S_{a_{k-2}a_{k-1}} \sum_{b \in [n]} S_{a_{k-1}b}S_{ba_1},
\]
and (iv) follows since $\sum_{b \in [n]} S_{a_{k - 1}b} S_{ba_1}$ is the $(a_{k-1},a_1)$ entry of $S^2 = (n - 1) I$.
Finally, in the case where $a_1 = a_{k - 1}$, we have
\[
\Delta(a_1, \ldots, a_k) = S_{a_1a_2} \cdots S_{a_{k-2}a_1}S_{a_1a_k}S_{a_ka_1},
\]
and (v) follows since $S_{a_1a_k}S_{a_ka_1} = 1$ provided $a_1 \neq a_k$.
\end{proof}

Let $\pi$ be a partition of $[k]$.  Recall that for $j \in [k]$, we let $\pi(j)$ denote the block of $\pi$ containing $j$.  We extend this notation to any integer $j$ by considering $\pi(j)$ to be the block of $\pi$ containing a representative of $j$ modulo $k$.  For convenience, we record the following immediate consequence of Lemma~\ref{deltaLemma}(iii).

\begin{lemma}\label{cycleLemma}
Let $\pi$ be a partition of $[k]$ and fix $j \in \mathbb{Z}$.  Define $\pi'$ to be the partition of $[k]$ with $\pi'(i) = \pi(i - j)$ for all $i \in [k]$.  Then
$
	V_n(\pi') = V_n(\pi).
$
\end{lemma}

To establish Lemma~\ref{lem.key lemma}(i), we will show separately that $V_n(\pi) \rightarrow 0$ for every crossing partition $\pi \in \Pi(k,t)$ and that $V_n(\pi) \rightarrow 0$ for every non-crossing partition $\pi \in \Pi(k,t)$ such that $G_\pi$ contains an odd cycle. 


\begin{lemma}\label{crossingLemma}
Let $\pi \in \Pi(k,t)$ be a crossing partition.  Then $V_n(\pi) \rightarrow 0$.
\end{lemma}

\begin{proof}
For $\pi \in \Pi(k,t)$ to be a crossing partition, it must hold that $t \geq 2$ and $k \geq 4$.  Observe that the case $t = 2$ follows immediately from Lemma~\ref{fewblocksLemma} since $k \geq 4$.
Now consider $t>2$ and suppose the lemma has been established for every crossing partition on $t - 1$ blocks.  By Lemma~\ref{fewblocksLemma}, we may further suppose that $k$ satisfies $t \geq k/2 + 1$. Then for $\pi \in \Pi(k,t)$, the pigeonhole principle guarantees that $\pi$ contains a singleton block $\{j\} \in \pi$.  By Lemma~\ref{cycleLemma}, we may assume $\{k\} \in \pi$.
We proceed in cases:

\textbf{Case I:} $\pi(1) = \pi(k - 1)$.  
We may apply Lemma~\ref{deltaLemma}(v) to obtain
\begin{align*}
V_n(\pi) &= \frac{1}{n^{k/2 + 1}} \sum_{a \in L_n(\pi)} \Delta(a(1), \ldots, a(k - 2), a(1), a(k))\\
&= \frac{1}{n^{k/2}} \sum_{a \in L_n(\pi \setminus \{k\})} \Delta(a(1), \ldots, a(k - 2)) + o(1).
\end{align*}
The restriction of $\pi \setminus \{k\}$ to $[k - 2]$ results in a crossing partition $\pi'$ of $[k - 2]$ into $t - 1$ blocks.  Moreover, the above expression for $V_n(\pi)$ implies
\[
V_n(\pi) = V_n(\pi') + o(1),
\]
and so our induction hypothesis provides $V_n(\pi) \rightarrow 0$.

\textbf{Case II:} $\pi(1) \neq \pi(k - 1)$.
Writing out $\pi = \{B_1, \ldots, B_{t - 1}, \{k\}\}$, we choose representatives $j_1, \ldots, j_{t - 1} \in [k - 1]$ with $\pi(j_i) = B_i$.  Then by Lemma~\ref{deltaLemma}(iv), we have
\begin{align*}
V_{n}(\pi)
&=\frac{1}{n^{k/2+1}}\sum_{a \in L_n(\pi \setminus \{k\})}
\sum_{\substack{a_k \in [n] \\ a_k \not \in a([k - 1])}} \Delta(a(1),\ldots,a(k-1),a_k)\\
&= -\frac{1}{n^{k/2 + 1}}\sum_{a \in L_n(\pi \setminus \{k\})} \sum_{i = 1}^{t - 1} \Delta(a(1), \ldots, a(k - 1), a(j_i)).
\end{align*}
For $i,j \in [t - 1]$, we define new blocks
\[
B_j^i := \begin{cases} B_j \cup \{k\} & \text{if } j = i \\ B_j & \text{if } j \neq i \end{cases}
\]
and the corresponding crossing partition $\pi^i = \{B_1^i, \ldots, B_{t - 1}^i\} \in \Pi(k,t-1)$.  Then
\[
V_n(\pi) = -\sum_{i = 1}^{t - 1} \frac{1}{n^{k/2 + 1}} \sum_{a \in L_n(\pi \setminus \{k\})}  \Delta(a(1), \ldots, a(k - 1), a(j_i)) = -\sum_{i = 1}^{t - 1} V_n(\pi^i).
\]
Since each $V_n(\pi^i) \rightarrow 0$ by our induction hypothesis, we see $V_n(\pi) \rightarrow 0$ as well.
\end{proof}

For non-crossing partitions, we will study the structure of the graph $G_\pi$ for $\pi \in \Pi(k,t)$, which we recall has vertex set $\pi$ and edges $\pi(j) \leftrightarrow \pi(j + 1)$ for all $j \in [k]$. Observe that if $G_\pi$ has a loop, then $\pi(j) = \pi(j + 1)$ for some $j \in [k]$, and so $V_n(\pi) = 0$ by Lemma~\ref{deltaLemma}(ii).  For this reason, we direct our attention to \textbf{loop-free} partitions $\pi$, that is, partitions $\pi$ for which $G_\pi$ is loop-free.

Given a loop-free graph $G$ on vertices $V$ with edges $E$, we say $v \in V$ is a \textbf{cut vertex} if the induced subgraph of $G$ on $V \setminus \{v\}$ is disconnected.
A graph with no cut vertices is called \textbf{biconnected}, and the \textbf{biconnected components} of a graph are its maximal biconnected subgraphs.
When the biconnected components of $G$ are all simple cycles, we call $G$ a \textbf{cactus}.

\begin{lemma}\label{cactusLemma}
If $\pi \in \Pi(k,t)$ is a loop-free non-crossing partition, then $t\geq k/2+1$ and $G_\pi$ is a cactus whose edges partition into $k - t + 1$ simple cycles.
\end{lemma}

\begin{proof}
First, suppose $G_\pi$ is a cactus whose edges partition into $k - t + 1$ simple cycles.
Since $G_\pi$ has no loops, the number of cycles is at least at least half the number of edges, that is, $k-t+1\geq k/2$.
Rearranging then gives $t\geq k/2+1$.
It remains to verify that $G_\pi$ is, indeed, a cactus whose edges partition into $k - t + 1$ simple cycles.

Fixing $k$, we proceed by induction on $k - t$.
If $\pi \in \Pi(k,k)$, then $G_\pi$ is itself a simple cycle and hence a cactus.
For $k - t > 0$, we now consider a loop-free non-crossing partition $\pi \in \Pi(k,t)$.
By the pigeonhole principle, we may select $B \in \pi$ such that $B$ contains at least two elements of $[k]$.
Let $j$ denote the least element of $B$.
Writing $B = \{j\} \sqcup B'$, we consider $\pi' \in \Pi(k,t + 1)$ defined by $\pi' = (\pi \setminus B) \cup \{B', \{j\}\}$.
Since $\pi'$ is also loop-free and non-crossing, our induction hypothesis guarantees that $G_{\pi'}$ is a cactus with $k - t$ simple cycles.
Our task is to use this information to show that $G_\pi$ is a cactus with $k - t + 1$ simple cycles.

Suppose first that $\{j\}$ and $B'$ reside in the same simple cycle of $G_{\pi'}$.
Then the simple cycles of $G_{\pi'}$ not containing $\{j\}$ and $B'$ remain simple cycles and biconnected components of $G_{\pi}$.
Moreover, by identifying $\{j\}$ and $B'$, we see that the simple cycle of $G_{\pi'}$ containing $\{j\}$ and $B'$ corresponds to two simple cycles of $G_{\pi}$ sharing the cut vertex $B=\{j\}\sqcup B'$.
As such, $G_{\pi}$ contains $k - t + 1$ biconnected components, each of which is a simple cycle.

We now claim that $\{j\}$ and $B'$ must reside in the same simple cycle of $G_{\pi'}$, in which case we are done by the previous paragraph.
Suppose instead that there exists a cut vertex $X \in \pi' \setminus \{\{j\},B'\}$ that separates $\{j\}$ and $B'$ within $G_{\pi'}$, and select $j' \in [k]$ with $\pi(j') = B'$.
Since $j$ is the least element of $B$, we necessarily have $j < j'$.
Furthermore, since $X$ separates $\{j\}$ and $B'$, we can traverse along a trail in $G_{\pi'}$ from $\{j\}$ to $X$, to $B'$, and back to $X$ to obtain indices $i \in (j,j')$ and $i' \in (j',k+j)$ with $\pi'(i) = \pi'(i') = X$.
These indices $j < i < j' < i'$ satisfy $\pi(j) = \pi(j')$ and $\pi(i) = \pi(i')$, contradicting our assumption that $\pi$ is non-crossing.
Hence, $\{j\}$ and $B'$ must reside in the same simple cycle of $G_{\pi'}$ as claimed.
\end{proof}

\begin{lemma}
\label{lem.cactus lemma}
For every loop-free non-crossing $\pi\in\Pi(k,t)$, each of the following holds:
\begin{itemize}
\item[(i)]
If $G_\pi$ contains any odd cycles, then $V_n(\pi) \rightarrow 0$.
\item[(ii)]
If the edges of $G_\pi$ partition into $m$ simple cycles of sizes $2s_1,\ldots,2s_m$, then
\[
V_n(\pi) \rightarrow (-1)^{k/2 - m}\cdot C_{s_1-1}\cdots C_{s_m-1}.
\]
\end{itemize}
\end{lemma}

\begin{proof}
Any loop-free non-crossing partition must have at least two blocks.
When $t = 2$, we may assume $k =2$ by Lemma~\ref{fewblocksLemma} so that the only partition under consideration is $\{\{1\},\{2\}\}$, in which case $m=1$ and $s_1=1$.
Lemma~\ref{deltaLemma}(i) allows us to verify the result in this case:
\[
V_n(\{\{1\},\{2\}\})
= \frac{1}{n^2} \sum_{\substack{a : [2] \rightarrow [n] \\ a(1) \neq a(2)}} \Delta(a(1),a(2))
= \frac{n(n - 1)}{n^2}
\rightarrow 1
= (-1)^{k/2-m}\cdot C_{s_1-1}.
\]
Now consider $t>2$, and suppose the lemma has been established for every loop-free non-crossing partition on $t-1$ blocks.
By Lemma~\ref{cactusLemma}, we may assume that $k$ satisfies $t\geq k/2+1$.
Then for $\pi\in\Pi(k,t)$, the pigeonhole principle guarantees that $\pi$ contains a singleton block $\{j\}\in\pi$.
By Lemma~\ref{cycleLemma}, we may assume $\{k\}\in\pi$.
We proceed in cases:

\textbf{Case I:} $\pi(1)=\pi(k-1)$.
We may apply Lemma~\ref{deltaLemma}(v) to obtain
\begin{align*}
V_n(\pi) &= \frac{1}{n^{k/2 + 1}} \sum_{a \in L_n(\pi)} \Delta(a(1), \ldots, a(k - 2), a(1), a(k))\\
&= \frac{1}{n^{k/2}} \sum_{a \in L_n(\pi \setminus \{k\})} \Delta(a(1), \ldots, a(k - 2)) + o(1).
\end{align*}
The restriction of $\pi \setminus \{k\}$ to $[k - 2]$ results in a loop-free non-crossing partition $\pi'$ of $[k - 2]$ into $t - 1$ blocks.
Moreover, the above expression for $V_n(\pi)$ implies
\begin{equation}
\label{eq.pinch loop free}
V_n(\pi) = V_n(\pi') + o(1).
\end{equation}
For (i), observe that if $G_\pi$ contains any odd cycles, then $G_{\pi'}$ must also contain odd cycles.
In this case, we may apply our induction hypothesis to $V_n(\pi')$ to conclude $V_n(\pi) \rightarrow 0$.
For (ii), the edges of $G_\pi$ partition into $m$ simple cycles of sizes $2s_1, \ldots, 2s_{m}$ with $s_m=1$, and so the edges of $G_{\pi'}$ partition into $m - 1$ simple cycles of sizes $2s_1, \ldots, 2s_{m - 1}$.
Then \eqref{eq.pinch loop free} and our induction hypothesis together imply
\[
V_n(\pi) \rightarrow (-1)^{k/2 - m}\cdot C_{s_1 - 1} \cdots C_{s_{m - 1} - 1}.
\]
Since $C_0 = 1$, this establishes (ii).

\textbf{Case II:} $\pi(1)\neq\pi(k-1)$.
In this case, $\pi(k)$ necessarily resides in a cycle of length $\ell\geq3$.
Select representatives $k,j_2,\ldots,j_t\in[k]$ with $\pi(k)=B_1$ and $\pi(j_i)=B_i$ so that the vertices in the cycle are given by $B_1,\ldots,B_\ell$.
Then we may apply Lemma~\ref{deltaLemma}(iv) to obtain
\begin{align*}
V_n(\pi) &= \frac{1}{n^{k/2 + 1}} \sum_{a \in L_n(\pi \setminus \{k\})} \sum_{\substack{a_k \in [n] \\ a_k \not \in a([k-1])}} \Delta(a(1), \ldots, a(k - 1), a_k)\\
&= -\frac{1}{n^{k/2 + 1}} \sum_{a \in L_n(\pi \setminus \{k\})} \sum_{i=2}^{t} \Delta(a(1), \ldots, a(k - 1), a(j_i)).
\end{align*}
For $i,j \in [2,t]$, define new blocks
\[
B_j^i = \begin{cases} B_j \cup \{k\} & \text{if } j = i \\ B_j & \text{if } j \neq i \end{cases}
\]
and the corresponding partitions $\pi^i = \{B_2^i, \ldots, B_t^i\}$, we have
\[
V_n(\pi) = -\sum_{i=2}^{t} V_n(\pi^i).
\]
By Lemma~\ref{crossingLemma}, $V_n(\pi^i) \rightarrow 0$ whenever $\pi^i$ is a crossing partition.  Since $\pi^i$ is obtained from $\pi$ by merging blocks $B_i$ and $\{k\}$, we can argue as in the proof of Lemma~\ref{cactusLemma} to conclude that $\pi^i$ is crossing if and only if $B_i$ and $\{k\}$ do not reside in the same simple cycle of $G_\pi$.  Hence,
 \[
 V_n(\pi) = -\sum_{i = 2}^{\ell} V_n(\pi^i) + o(1),
 \]
where each $\pi^i$ is non-crossing for $2 \leq i \leq \ell $.  Both $\pi^2$ and $\pi^{\ell}$ contain loops, so $V_n(\pi^2) = V_n(\pi^{\ell}) = 0$ by Lemma~\ref{deltaLemma}(ii).
When $\ell = 3$, this gives $V_n(\pi) \rightarrow 0$, as desired by (i).
Supposing for the remainder that $\ell \geq 4$, we must still compute the limit of
\begin{equation}\label{suminCycle}
V_n(\pi) = -\sum_{i = 3}^{\ell - 1} V_n(\pi^i) + o(1).
\end{equation}
Observe that our cycle $\{B_1, \ldots, B_\ell\}$ in $G_\pi$ of length $\ell$ corresponds to the two simple cycles in $G_{\pi^i}$ of $\{B_2^i, \ldots, B_i^i\}$ and $\{B_i^i, B_{i + 1}^i, \ldots, B_{\ell}^i\}$ with lengths $i-1$ and $\ell - i+1$ and share the cut vertex $B_i^i$.  Moreover, all other simple cycles are identical between the two graphs.

If $\ell$ is odd, then for each $i \in [3,\ell-1]$, either $i-1$ or $\ell - i+1$ is odd, and so $G_{\pi^i}$ must have an odd cycle.  Since each $\pi^i$ has $t - 1$ blocks and an odd cycle, we can apply our induction hypothesis to conclude that each $V_n(\pi^i) \rightarrow 0$ so that $V_n(\pi) \rightarrow 0$, as desired by (i).  Suppose instead that $\ell$ is even, but $G_\pi$ has an odd cycle.  This odd cycle is also contained in each $G_{\pi^i}$ for $i \in [3,\ell - 1]$, and again we can apply our induction hypothesis to conclude that $V_n(\pi) \rightarrow 0$, thereby establishing (i).

Finally, for (ii), suppose that $\ell$ is even and that the edges of $G_\pi$ partition into $m$ cycles of lengths $2s_1, \ldots, 2s_m$ with $2s_m = \ell$.  Notice that if $i-1$ is odd, then $G_{\pi^i}$ contains an odd cycle, and $V_n(\pi^i) \rightarrow 0$.  Since the contribution of these terms is negligible, we must compute the limit of
\[
V_n(\pi) = -\sum_{i = 1}^{\ell/2 - 1} V_n(\pi^{2i+1}) + o(1).
\]
The cycles of lengths $2s_1, \ldots, 2s_{m - 1}$ are common to both $G_\pi$ and $G_{\pi^{2i+1}}$, while the cycle of length $\ell = 2s_m$ in $G_\pi$ corresponds to two cycles of length $2i$ and $\ell - 2i$ in $G_{\pi^{2i+1}}$.  Applying our induction hypothesis, we have
\[
V_n(\pi) \rightarrow (-1)^{k/2 - m} \cdot C_{s_1 - 1} \cdots C_{s_{m - 1} - 1} \sum_{i = 1}^{\ell/2 - 1} C_{i - 1}C_{\ell/2 - i - 1}.
\]
Reindexing and applying the convolution identity for Catalan numbers, we have
\[
\sum_{i = 1}^{\ell/2 - 1} C_{i - 1}C_{\ell/2 - i - 1} = \sum_{i = 0}^{\ell/2 - 2} C_i C_{\ell/2 - i - 2} = C_{\ell/2 - 1}.
\]
Hence, $V_n(\pi) \rightarrow (-1)^{k/2 - m} \cdot C_{s_1 - 1} \cdots C_{s_m - 1}$, thereby establishing (ii).
\end{proof}

\begin{proof}[Proof of Lemma~\ref{lem.key lemma}]
To prove (i), consider $\pi \in \Pi(k,t) \setminus \operatorname{EC}(k,t)$.
Then either $\pi$ is crossing, $G_\pi$ contains a loop, or $G_\pi$ contains an odd cycle.
If $\pi$ is crossing, then $V_n(\pi) \rightarrow 0$ by Lemma~\ref{crossingLemma}.
If $G_\pi$ contains a loop, then $V_n(\pi) = 0$ by Lemma~\ref{deltaLemma}(ii).
If $\pi$ is loop-free and non-crossing but $G_\pi$ contains an odd cycle, then $V_n(\pi) \rightarrow 0$ by Lemma~\ref{lem.cactus lemma}(i).
This establishes~(i).
Finally, (ii) follows from applying both Lemmas~\ref{cactusLemma} and~\ref{lem.cactus lemma}(ii).
\end{proof}

\section*{Acknowledgments}
MM and DGM were partially supported by AFOSR FA9550-18-1-0107.
DGM was also supported by NSF DMS 1829955 and the Simons Institute of the Theory of Computing.

\end{document}